\documentclass{cccg14arxiv}

\usepackage{amsmath}
\usepackage{amssymb}
\usepackage{balance}
\usepackage{dsfont}
\usepackage{graphicx}
\usepackage[utf8x]{inputenc}
\usepackage[disable]{todonotes}
\usepackage{xspace}
\usepackage[pdfborder={0 0 0}]{hyperref}
\usepackage{url}

\title{A PTAS for the continuous 1.5D Terrain Guarding Problem}

\newcommand*\samethanks[1][\value{footnote}]{\footnotemark[#1]}
\author{%
	Stephan~Friedrichs\thanks{Department of Computer Science, TU Braunschweig, Germany.
		{\tt stephan.friedrichs@tu-bs.de, mhsaar@gmail.com, c.schmidt@tu-bs.de}}
	\and Michael~Hemmer\samethanks
	\and Christiane~Schmidt\samethanks}

\index{Friedrichs, Stephan}
\index{Hemmer, Michael}
\index{Schmidt, Christiane}

\newtheorem{definition}[theorem]{Definition}

\newcommand{\bigO}{\operatorname{O}}

\newcommand{\NP}{\textit{NP}}

\newcommand{\ie}{i.\,e.\xspace}

\newcommand{\opt}{\operatorname{OPT}}

\newcommand{\tgp}{\operatorname{TGP}}
\newcommand{\V}{\operatorname{\mathcal{V}}}
\newcommand{\VK}{\operatorname{VK}}
\newcommand{\wrt}{w.\,r.\,t.\xspace}

\begin{document}

\thispagestyle{empty}
\maketitle

\begin{abstract}
In the continuous 1.5-dimensional terrain guarding problem we are given an $x$-monotone chain (the \emph{terrain} $T$) and ask for the minimum number of point guards (located anywhere on $T$), such that all points of $T$ are covered by at least one guard.
It has been shown that the 1.5-dimensional terrain guarding problem is \NP-hard.
The currently best known approximation algorithm achieves a factor of $4$.
For the discrete problem version with a finite set of guard candidates and a finite set of points on the terrain that need to be monitored, a polynomial time approximation scheme (PTAS) has been presented~\cite{gkkv-aasftg-09}.
We show that for the general problem we can construct finite guard and witness sets, $G$ and $W$, such that there exists an optimal guard cover $G^* \subseteq G$ that covers $T$, and when these guards monitor all points in $W$ the entire terrain is guarded.
This leads to a PTAS as well as an (exact) IP formulation for the continuous terrain guarding problem.
\end{abstract}

\section{Introduction}
\label{sec:introduction}

Let a \emph{terrain} $T$ denote an $x$-monotone chain defined by its vertices $V = \{ v_1, \dots, v_n \}$.
It has edges $E = \{ e_1, \dots, e_{n-1} \}$ with $e_i = \overline{v_i v_{i+1}}$.
Due to its monotonicity the points on $T$ are totally ordered with regard to their $x$-coordinate.
For $p,q \in T$, we write $p < q$ if $p$ is \emph{left} of $q$, \ie, if $p$ has a smaller $x$-coordinate than $q$.

A point $p \in T$ \emph{sees} or \emph{covers} $q \in T$ iff $\overline{pq}$ is nowhere below $T$.
$\V(p)$ is the \emph{visibility region} of $p$ with $\V(p) = \{ q \in T \mid \textnormal{$p$ sees $q$} \}$.
%
$\V(p)$ is not necessarily connected, and can be considered as the union of $\bigO(n)$ maximal subterrains, compare Figure~\ref{fig:visibility_region}.
We say that $q \in \V(p)$ is \emph{extremal} in $\V(p)$, if $q$ has a maximal
or minimal $x$-coordinate within its connected component of $\V(p)$.

\begin{figure}
  \begin{center}
    \includegraphics[width=0.45\textwidth]{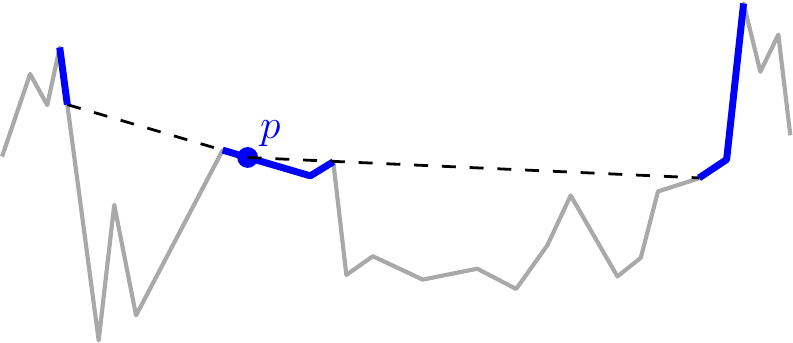}
  \end{center}
  \caption{The visibility region $\V(p)$ (blue) of point $p$.}
  \label{fig:visibility_region}
\end{figure}

For $G \subseteq T$ we abbreviate $\V(G) := \bigcup_{g \in G} \V(g)$.
A set $G \subseteq T$ with $\V(G) = T$ is named a \emph{(guard) cover} of $T$.
In this context, $g \in G$ is referred to as \emph{guard}.


\begin{definition}[Terrain Guarding Problem]\label{def:tgp}
	In the \emph{Terrain Guarding Problem (TGP)}, we are given a terrain $T$ and sets of guard candidates and witnesses $G, W \subseteq T$.
	We seek a minimum-cardinality guard cover $G^* \subseteq G$ such that $W \subseteq \V(G^*)$, \ie, $G^*$ covers $W$, and abbreviate this problem by $\tgp(G,W)$.
	We assume $W \subseteq \V(G)$, \ie, that $\tgp(G,W)$ always has a feasible solution.
\end{definition}

Note that $\tgp(T,T)$ is the continuous TGP and $\tgp(V,T)$ is the special case with vertex guards.

\subsection{Related Work}
\label{sec:rw}

The terrain guarding problem is closely related to the well known \emph{Art Gallery Problem} (AGP) where the objective is to find a minimum cardinality guard set that allows complete coverage of a polygon $P$.
Many different versions of this problem have been considered, including variants with guards restricted to be located on vertices (\emph{vertex guards}), patrolling along edges (\emph{edge guards}), or located on arbitrary positions in $P$ (\emph{point guards}); each in simple polygons and in polygons with holes.

The first result was obtained by Chv\'{a}tal~\cite{c-actpg-75}, who proved the ``Art Gallery Theorem'', answering a question posed by Victor Klee in 1973: $\lfloor \frac{n}{3} \rfloor$ many guards are always sufficient and sometimes necessary to guard a polygon of $n$ vertices.
A simpler and elegant proof of the sufficiency was later given by Fisk~\cite{f-spcwt-78}.
Related results were obtained for different polygon classes.
The optimization problem was shown to be NP-hard for various problem versions~\cite{os-snpdp-83,sh-tnhag-95}, even the allegedly easier problem of finding a minimum cardinality vertex guard set in simple polygons is NP-hard~\cite{ll-ccagp-86}.
Eidenbenz et al.~\cite{esw-irgpt-01} gave bounds on the approximation ratio: For polygons with holes a lower bound of $\Omega(\log n)$ holds, for vertex, edge and point guards in simple polygons they showed the problem to be APX-hard.

For detailed surveys on the AGP see O'Rourke~\cite{r-agta-87} or Shermer~\cite{s-rrag-92} for classical results.
\todo[inline]{or~\cite{} for more recent computational developments.}

Motivation for terrain guarding is the placement of street lights or security cameras along roads~\cite{gkkv-aasftg-09}, or the optimal placement of antennas for line-of-sight communication networks~\cite{bkm-acfaafotg-07}.

For the terrain guarding problem the focus was on approximation algorithms, because NP-hardness was generally assumed, but had not been established by then.
The first who were able to establish a constant-factor approximation algorithm were Ben-Moshe et al.~\cite{bkm-acfaafotg-07}.
They presented a combinatorial constant-factor approximation for the discrete vertex guard problem version $\tgp(V,V)$, where only vertex guards are used to cover only the vertices, and were able to use it as a building block for an $\bigO(1)$-approximation of the continuous terrain guarding variant $\tgp(T,T)$.
The approximation factor of this algorithm was never stated by the authors, but was claimed to be 6 in~\cite{k-a4aafgdt-06} (with minor modifications).
Another constant-factor approximation based on $\epsilon$-nets and Set Cover was given by Clarkson and Varadarajan~\cite{cv-iaags-07}.
King~\cite{k-a4aafgdt-06} presented a 4-approximation (which was later shown to actually be a 5-approximation~\cite{k-eaag}), both for the discrete $\tgp(V,V)$ and the continuous $\tgp(T,T)$ problem.
The most recent $\bigO(1)$-approximation was presented by Elbassioni et al.~\cite{emms-iafgdt-08,ekmms-iag15-11}:
For non-overlapping discrete sets $G,W \subset T$ LP-rounding techniques lead to a 4-approximation
(5-approximation if $G \cap W \neq \emptyset$) for $\tgp(G,W)$ as well as for
the continuous case $\tgp(T,T)$.
This approximation is also applicable for the more general \emph{weighted} terrain guarding problem:
Weights are assigned to the guards and a minimum weight guard set is to be identified.

Finally, in 2009, Gibson et al.~\cite{gkkv-aasftg-09} showed that the discrete terrain guarding problem allows a \emph{polynomial time approximation scheme} (PTAS) based on local search.
They present PTASs for two problem variants: for $\tgp(G,W)$ where $G$ and $W$ are (not necessarily disjoint) finite subsets of the terrain $T$ and for $\tgp(G,T)$, \ie, the variant with a finite guard candidate set $G$.
For the continuous case, \ie, $\tgp(T,T)$, they claim that the local search works as well, but that they were not yet able to limit the number of bits needed to represent the guards maintained by the local search.
Thus, to the best of our knowledge, no PTAS for $\tgp(T,T)$ has been established until now.

The NP-hardness of the TGP was settled after all these approximation results: King and Krohn~\cite{kk-tginph-10} proved both the discrete and the continuous case to be NP-hard by a reduction from PLANAR 3SAT in 2006.

Other problems considered in the context of terrains include, for example, guards that are allowed to ``hover'' over the terrain~\cite{e-aatg-02}, or the computation of visibility polygons, \ie, the set of points on the terrain visible to a point $p$ on the terrain~\cite{lss-facvm-14}.

\subsection{Our Contribution}
\label{sec:contribution}

We present a discretization, \ie, finite sets $G, W \subset T$ such that any optimal solution for $\tgp(G,W)$
is optimal and feasible for $\tgp(T,T)$ as well:
\begin{enumerate}
\item
	For the sake of completeness we argue that for each finite guard candidate set $G$ there exists a finite witness set $W(G)$, such that a solution for $\tgp(G,W(G))$ is feasible for $\tgp(G,T)$ (Section~\ref{sec:discrete-witnesses}).

\item
	For each terrain $T$ there is a finite guard candidate set $U$, such that for each (possibly optimal) guard cover $C \subset T$ for $\tgp(T,T)$ there exists a guard cover $C' \subseteq U$ with $|C'| = |C|$ (Section~\ref{sec:discrete-guards}).

\item
	It then follows that any feasible optimal solution of $\tgp(U,W(U))$ is also optimal and feasible for $\tgp(T,T)$ (Section~\ref{sec:compl}).

\item
	Combining this discretization with the PTAS of Gibson et al.~\cite{gkkv-aasftg-09} for the discrete $\tgp(G,W)$ case, we obtain a PTAS for $\tgp(T,T)$ (Section~\ref{sec:ptas}).

\item
	The discretization also yields an IP formulation for exact solutions (Section~\ref{sec:exact}).
\end{enumerate}

\section{Discretization}
\label{sec:discretization}

This section is our core contribution.
We consider the following problem:
Given a terrain $T$, construct finite sets $G, W \subset T$ (guard candidate and witness points), such that any optimal (minimum-cardinality) solution for $\tgp(G,W)$ is optimal for $\tgp(T,T)$ as well.

We achieve this in three steps.
(1)~Provided with some finite guard candidate set $G \subset T$, Section~\ref{sec:discrete-witnesses} shows how to construct a finite witness set $W(G)$ from it, such that any feasible solution of $\tgp(G, W(G))$ is feasible for $\tgp(G,T)$ as well.
(2)~Section~\ref{sec:discrete-guards} discusses a finite set of guards $U$ that allows minimum-cardinality coverage of $T$.
(3)~In Section~\ref{sec:compl}, we argue that $\tgp(T,T)$ can be optimally solved using $\tgp(U,W(U))$, which is useful for a PTAS (Section~\ref{sec:ptas}) as well as for exact solutions (Section~\ref{sec:exact}).

Note that a discretization similar to that of Chwa et~al.\ for polygons~\cite{cjkmos-gagbgw-05} does not work.
Chwa et~al.\ pursue the idea of \emph{witnessable} polygons, which allow placing a finite set of witnesses, such that covering the witnesses implies full coverage of the polygon.
They show that this is possible if the polygon can be covered by a finite set of visibility kernels, \ie, if no point has a point-shaped visibility kernel.
Unfortunately, we can easily construct a terrain $T$ not coverable by a \emph{finite} union of \emph{visibility kernels}
$\VK(w) = \{ w' \in T \mid \V(w) \subseteq \V(w') \}$,
\ie, $T$ with finite $\left| \VK(w) \right|$ for some $w \in T$, compare Figure~\ref{fig:vis-kernel}.

\begin{figure}
	\begin{center}
		\includegraphics[width=0.45\textwidth]{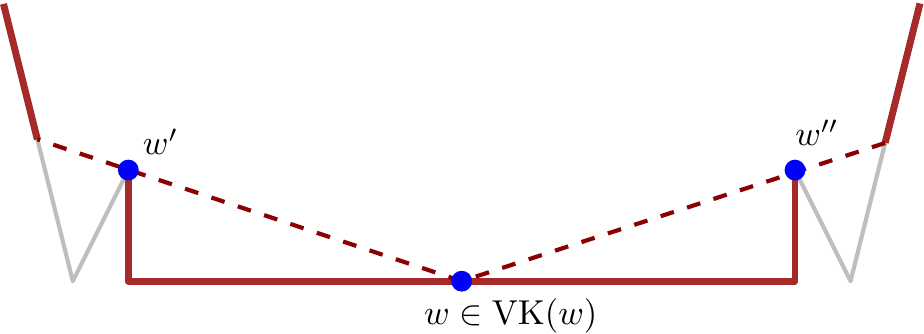}
	\end{center}
	\caption{Witness $w$, $\V(w)$ highlighted in red, and finite visibility kernel $\VK(w) = \{w, w', w''\}$ marked in blue.}
	\label{fig:vis-kernel}
\end{figure}

\subsection{Witnesses}
\label{sec:discrete-witnesses}

Suppose we are given a terrain $T$ and a finite set $G \subset T$ of guard candidates
with $\V(G) = T$ and we want to cover $T$ using only guards $C \subseteq G$, \ie, we want to solve $\tgp(G,T)$.
$G$ could be the set $V$ of vertices to solve the vertex guard variant of the TGP or any other finite set,
especially the one in Equation~\eqref{eq:u} of Section~\ref{sec:discrete-guards},
which contains all guard candidates necessary to find an optimal solution of the
continuous version of the problem, $\tgp(T,T)$.

$G$ is finite by assumption, but $T$ isn't, so we generate a finite set $W(G) \subset T$ of witness points, such that any feasible solution for $\tgp(G,W(G))$ also is feasible for $\tgp(G,T)$.

Let $g \in G$ be one of the guard candidates.
$\V(g)$ subdivides $T$ into $\bigO(n)$ closed subterrains.
For the sake of simplicity, we project those subterrains onto
the $x$-axis, allowing us to represent $\V(g)$ as a set of closed intervals.

We consider the overlay of all visibility intervals of all guard candidates in $G$, see Figure~\ref{fig:overlay} for an overlay of two guards.
It forms a subdivision consisting of \emph{maximal intervals} and \emph{end points}.
Every point in a \emph{feature} $f$ (either end point or maximal interval) of the subdivision is seen by the same set of guards
\begin{equation}
	G(f) = \left\{ g \in G \mid f \subseteq \V(g) \right\}.
\end{equation}

It is possible to simply place one witness in every feature of the subdivision.
Covering all $\bigO(n \cdot |G|)$ witnesses implies full coverage of $T$.

Similar to the shadow atomic visibility polygons in~\cite{crs-aaafmvgoag-11},
we can further reduce the number of witnesses by only using those
features $f$ with inclusion-minimal $G(f)$:
\todo{Maybe replace by survey cite later}

\begin{figure}
  \begin{center}
    \includegraphics[width=0.45\textwidth]{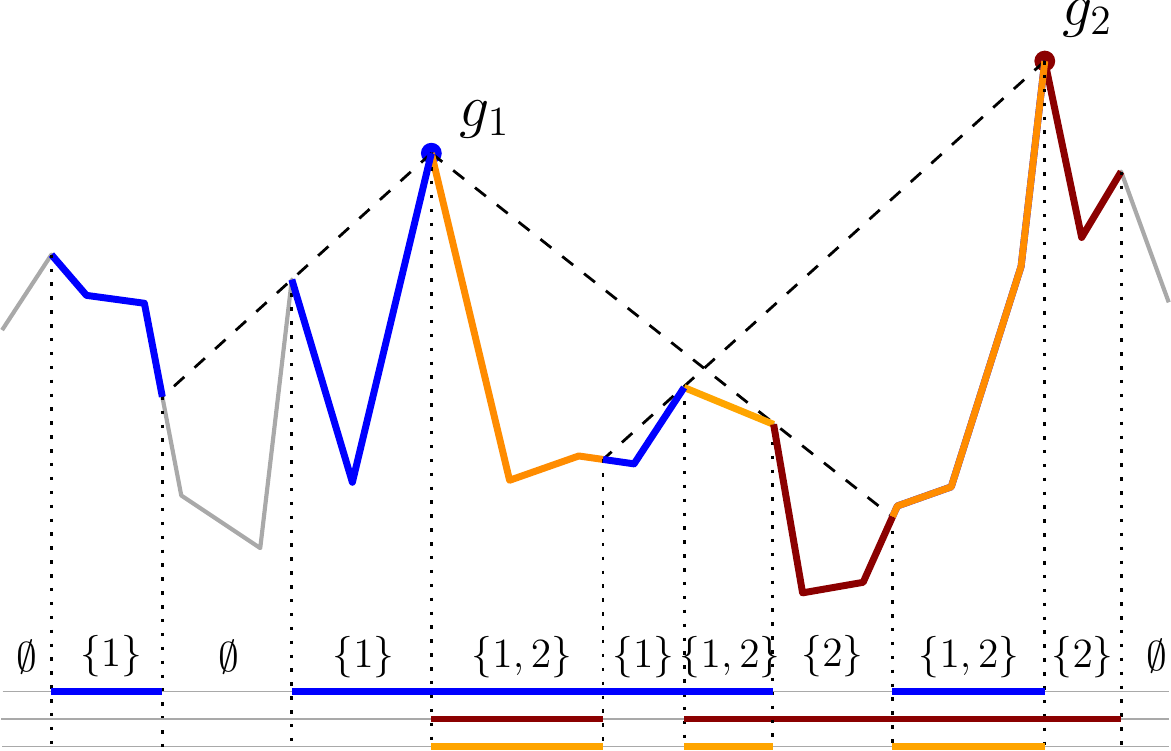}
  \end{center}
  \caption{Visibility overlay of two guards $g_1$ and $g_2$.
    $\V(g_1)$ is indicated in blue; $\V(g_2)$ in
    red; $\V(g_1) \cap $$\V(g_2)$ in orange. }
  \label{fig:overlay}
\end{figure}

\begin{theorem}\label{thm:w}
  Consider a terrain $T$ and a finite set of guard candidates $G$ with $\V(G) = T$.
  Let $F$ denote the features of the visibility overlay of $G$ and $w_f \in f$ an arbitrary point in the feature $f \in F$.
  Then for
  \begin{equation}\label{eq:w}
    W(G) = \left\{ w_f \mid f \in F, \, \textnormal{$G(f)$ is inclusion-minimal} \right\}
  \end{equation}
  any feasible solution of $\tgp(G, W(G))$ is feasible for $\tgp(G,T)$.
\end{theorem}

\begin{proof}
  Let $C \subseteq G$ be a feasible cover of $W(G)$ and suppose some
  point $w \in T$ is not covered by $C$.
  By assumption, some point in $G$ can see $w$, so $w$ must be part
  of some feature $f$ of the visibility overlay of $G$.
  $W(G)$ either contains some witness in $w_f \in f$ or a witness
  $w_{f'}$ with $G(f') \subseteq G(f)$ by construction.
  In the first case, $w$ must be covered, otherwise $w_f$ would
  not be covered and $C$ would be infeasible for $\tgp(G,W(G))$.
  In the second case $w_{f'}$ is covered, so some guard in $G(f')$
  is part of $C$, but that guard also covers $f$ and therefore $w$, a contradiction.
\end{proof}

\begin{figure*}
  \begin{center}
    \includegraphics[width=0.7\textwidth]{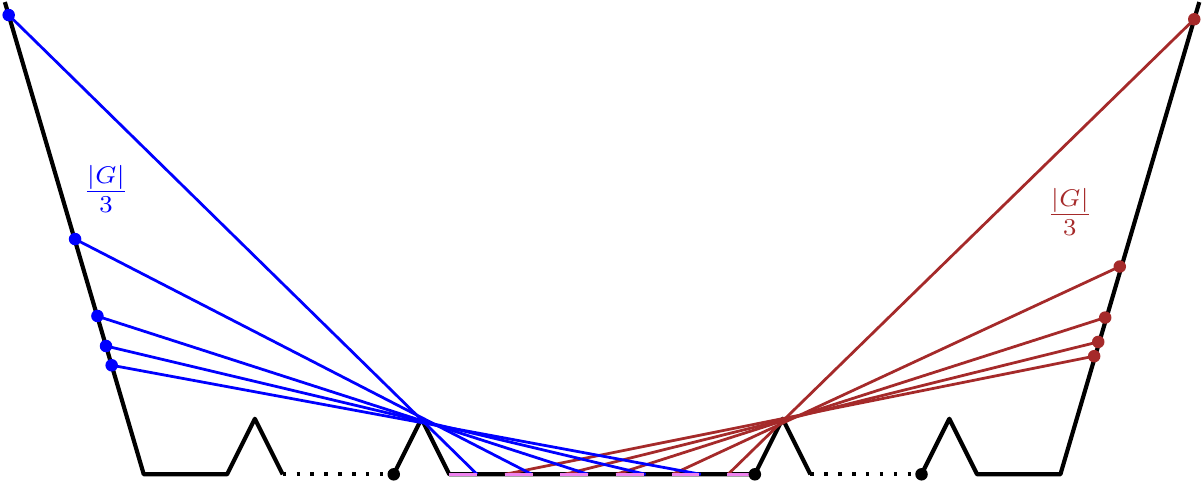}
  \end{center}
  \caption{
    The set of inclusion-minimal features may still be $\bigO(n \cdot |G|)$.
    For $n\in \bigO(|G|)$ consider the terrain with $\bigO(n)$ valleys, with
    $|G|/3$ guards placed on both slopes (blue on the left, red on the right).
    In addition there is one black guard placed in each valley.
    Thus, each of the $\bigO(n)$ valleys contains $\bigO(G)$ inclusion-minimal
    intervals depicted in violet.
  }
  \label{fig:inclusion_minimal}
\end{figure*}

\begin{obs}\label{obs:witcard}
Using only the set of witnesses on inclusion-minimal features may not reduce the worst case complexity of $\bigO(n \cdot |G|)$ witnesses, see Figure~\ref{fig:inclusion_minimal}.
\end{obs}
Nevertheless, we expect inclusion-minimal witnesses to speed up an implementation.

\begin{obs}
  $W(G)$ does not need to contain any interval end point.
  An end point $p$ with adjacent maximal intervals
  $I_1, I_2$ can always be left out, because $G(p) = G(I_1) \cup G(I_2)$.
\end{obs}

\subsection{Guard Positions}
\label{sec:discrete-guards}

Throughout this section, let $T$ be a terrain, $V$ its vertices and $E$ its edges;
let $C \subset T$ be some finite, not necessarily optimal, guard cover of $T$.
Moreover, let $U$ be the union of $V$ with all $x$-extremal points of all visibility
regions of all vertices:
\begin{equation}\label{eq:u}
  U := V \cup \bigcup_{v \in V} \left\{ p \mid  \textnormal{$p$ is extremal in $\V(v)$} \right\}.
\end{equation}

\begin{obs}\label{obs:guardcard}
The set $U$ has cardinality $\bigO(n^2)$, as noted by Ben-Moshe et~al.~\cite{bkm-acfaafotg-07}.
\end{obs}

Ben-Moshe et~al. also add an arbitrary point of $T$ between each pair of consecutive points in $U$.
They use this extended set as their witness set.
We, however, show in this section that $U$ is sufficient to admit an optimal guard cover for $T$.

Our basic idea is that for any cover $C$ it is always possible to move guards in $C\setminus U$ to a neighboring point in $U$ without losing coverage.
In particular, this is possible for an optimal guard cover.

First observe that we can not lose coverage for an edge $e$ that is entirely covered by a guard $g \in C \setminus U$ if we move $g$ to one of its neighbors in $U$.

\begin{lemma}\label{lem:guard-edge-u}
  Let $g \in C \setminus U$ be a guard that covers an entire edge $e_i \in E$.
  Then $u_\ell, u_r$, the ``$U$-neighbors'' of $g$ with
  \begin{equation}\label{eq:ul-ur}
    \begin{aligned}
      u_\ell & = \max\{ u \in U \mid u < g \} \\
      u_r & = \min\{ u \in U \mid g < u \}
    \end{aligned}
  \end{equation}
  each entirely cover $e_i$, too.
\end{lemma}

\begin{proof}
  $g$ covers $e_i$, so $v_i, v_{i+1} \in \V(g)$, implying $g \in \V(v_i) \cap \V(v_{i+1})$.
  Moving $g$ towards $u_\ell$ does not move $g$ out of $\V(v_i)$ or $\V(v_{i+1})$, as the boundaries of those regions are contained in $U$ by construction.
  So $v_i, v_{i+1} \in \V(u_\ell)$ and $e_i \subseteq \V(u_\ell)$.
  Analogously: $e_i \subseteq \V(u_r)$.
\end{proof}

It remains to consider the \emph{critical edges}, \ie, those that are not entirely covered by a single guard, compare Figure~\ref{fig:critical_edge}.

\begin{figure}[b]
  \begin{center}
    \includegraphics[width=0.45\textwidth]{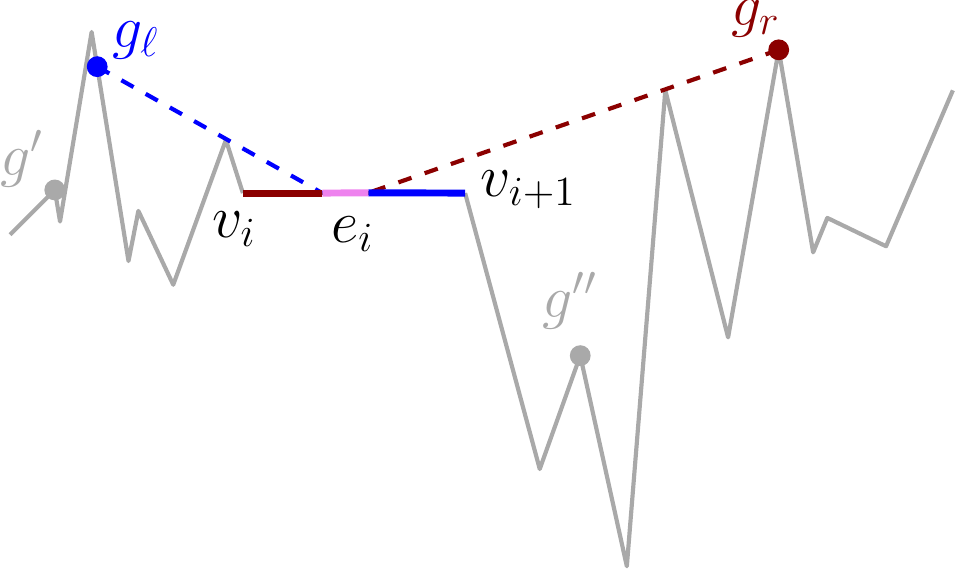}
  \end{center}
  \caption{The edge $e$ is critical w.r.t. the given guard cover.
    The right (left) part of $e_i$,
    indicated in blue (red), is seen by $g_\ell$ ($g_r$) only.
  }
  \label{fig:critical_edge}
\end{figure}

\begin{definition}[Critical Edge]
  $e \in E$ is a \emph{critical edge \wrt $g $} in the cover $C$ if
  $C \setminus \{ g \}$ covers some part of, but not all of, $e$.
\end{definition}
That is, after removing $g$, $e$ is only partially covered.

\begin{definition}[Left-Guard/Right-Guard]
  $g \in C$ is a \emph{left-guard (right-guard)} of
  $e_i \in E$ if $g < v_i$ ($v_{i+1} < g$) and $e_i$ is critical \wrt $g$.
  We call $g$ \emph{left-guard (right-guard)} if it is
  a left-guard (right-guard) of some $e \in E$.
\end{definition}

For the sake of completeness, we state and prove the following lemma, which also follows from the well-established order claim~\cite{bkm-acfaafotg-07}:

\begin{lemma}\label{lem:single-interval-vis}
  Let $g \in C$ be a left-guard (right-guard) of $e_i \in E$.
  Then $g$ covers a single interval of $e_i$, which includes $v_{i+1}$ ($v_i$).
\end{lemma}

\begin{proof}
  Refer to Figure~\ref{fig:critical_edge}.
  Obviously, $g$ is nowhere below the line supporting $e_i$.
  Let $p$ be a point on $e_i$ seen by $g$.
  It follows that $\overline{g p}$ and $\overline{p v_{i+1}}$
  form an $x$-monotone convex chain that is nowhere  below $T$.
  Thus, the secant $\overline{g v_{i+1}}$ is nowhere below $T$.
  It follows that $g$ sees $v_{i+1}$ as well as any
  point between $p$ and $v_{i+1}$ (same argument). The argument for the right-guard is analogous.
\end{proof}

\begin{cor}\label{cor:unique-left-guard}
  For each critical edge $e$ there is exactly one
  left-guard (right-guard) in $C$.
\end{cor}

\begin{cor}\label{cor:vis-intersect}
  Let $e\in E$ be a critical edge and $g_\ell, g_r \in C$ be its
  left- and right-guard. Then $\V(g_\ell)\cap e \cap \V(g_r) \not = \emptyset $.
\end{cor}

The following Lemma shows that we can move a guard $g \in C \setminus U$
to its left neighbor in $U$ without losing coverage of $T$ if
$g$ is not a right-guard.

\begin{lemma}\label{lem:left-guard-to-u}
  Let $C$ be some finite cover of $T$, $g \in C \setminus U$
  be a left- but no right-guard, and let $u_\ell$ be the
  left $U$-neighbor of $g$ as in Equation~\eqref{eq:ul-ur}.
  Then
  \begin{equation}
     C' = \left( C \setminus \{ g \} \right) \cup \{ u_\ell \}
   \end{equation}
   is a guard cover of $T$.
 \end{lemma}

\begin{proof}
  By Lemma~\ref{lem:guard-edge-u}, edges entirely covered by $g$ are also covered by $u_\ell$.
  So consider $p \in  \V(g)\cap e_r$ on a critical edge $e_r$ \wrt $g$ as
  depicted in Figure~\ref{fig:left-guard-to-u}.
  The guard $g$ is in the interior of an edge $e$ since $g \in C \setminus U$.
  As $p$ is seen by $g$ it must be nowhere below the line supporting $e$.
  It follows that segments $\overline{u_\ell g}$ and $\overline{gp}$ form an
  $x$-monotone convex chain that is nowhere below $T$.
  Hence, the secant $\overline{u_\ell p}$ is nowhere below $T$, so $u_\ell$ sees $p$.
  In particular,  it holds that
  \begin{equation}
    \V(g)\cap e_r \subseteq \V(u_\ell)\cap e_r,
  \end{equation}
  for every critical edge $e_r$ \wrt $g$.
\end{proof}

\begin{figure}
  \begin{center}
    \includegraphics[width=0.45\textwidth]{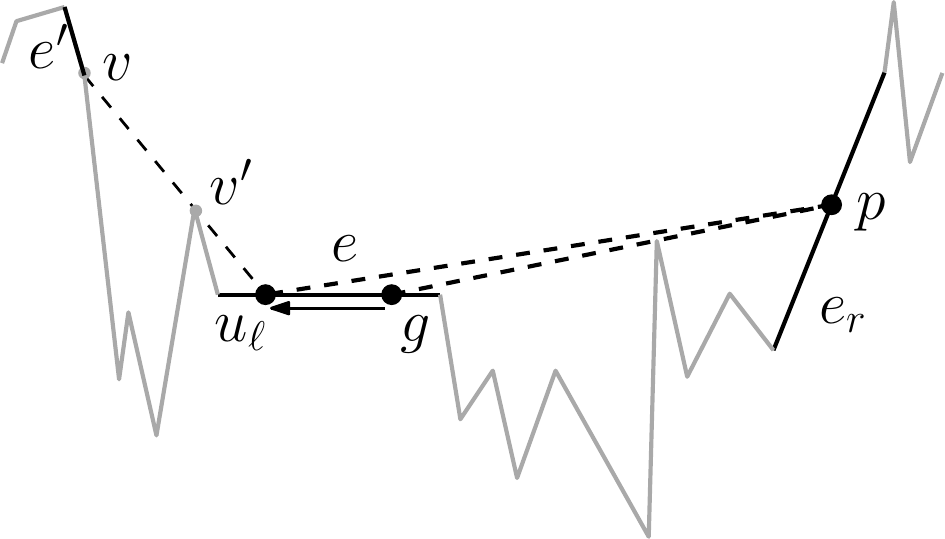}
  \end{center}
  \caption{Moving the left-guard $g$ further to the left.
    Any point $p$ on critical edge $e_r$ that is seen by $g$
    remains visible while moving $g$ to its left U-neighbor $u_\ell$.
    Also, non-critical edge $e'$ remains entirely visible since $g$
    does not cross $u_\ell$ which is induced by $v$ and $v'$.}
  \label{fig:left-guard-to-u}
\end{figure}

 \begin{cor}\label{cor:right-guard-to-u}
   Let $C$ be some finite cover of $T$, $g \in C \setminus U$
   be a right- but no left-guard, and let $u_r$ be the
   right $U$-neighbor of $g$ as in Equation~\eqref{eq:ul-ur}.
   Then
   \begin{equation}
     C' = \left( C \setminus \{ g \} \right) \cup \{ u_r \}
   \end{equation}
   is a guard cover of $T$.
 \end{cor}

\begin{figure}
  \begin{center}
    \includegraphics[width=0.45\textwidth]{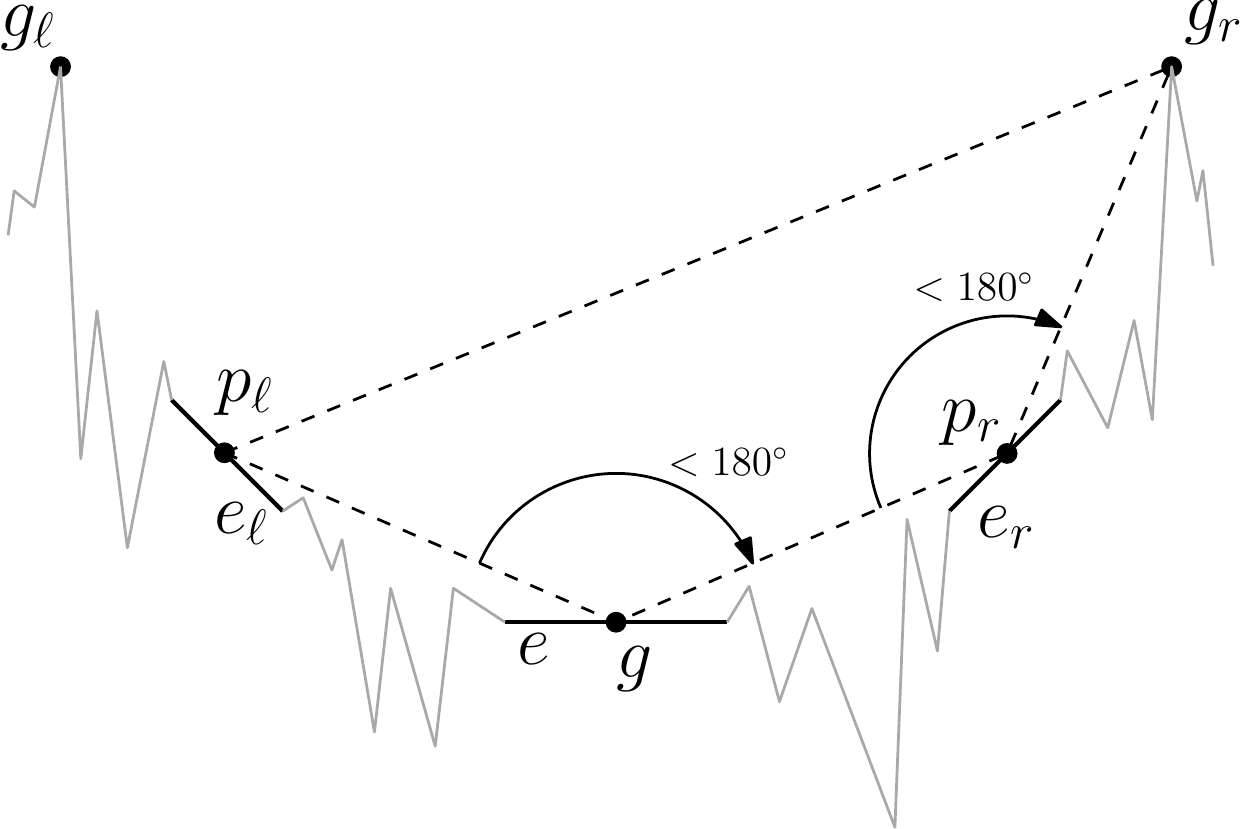}
  \end{center}
  \caption{No guard in $T \setminus U$ is left- and right-guard. Any point on the critical edge
    $e_\ell$ that is seen by $g$ is also seen by the guard $g_r$.
    Hence, $e_\ell$ can not be critical with respect to $g$, contradiction.
    A symmetric argument applies for $e_r$.}
  \label{fig:leftright-guard-to-u}
\end{figure}

\begin{lemma}\label{lem:leftright-guard-to-u}
  Let $C$ be some finite cover of $T$.
  No $g \in C \setminus U$ is both a left- and a right-guard.
\end{lemma}
\begin{proof}
  Refer to Figure~\ref{fig:leftright-guard-to-u}.
  We prove the claim by contradiction.
  Suppose that $g \in C \setminus U$ is
  the left-guard for $e_r$ (to the right of $g$) and
  the right-guard for $e_\ell$ (to the left of~$g$).

  Since $g$ is  the left-guard for critical edge $e_r$
  there must also be the right-guard $g_r$ for $e_r$.
  By Corollary~\ref{cor:vis-intersect} there is at least one point
  $p_r\in e_r$ that is seen by $g$ and $g_r$.

  Since $g \in C \setminus U$, it must be in the interior of some edge $e$.
  Now consider $p_r$ and any point $p_\ell \in \V(g)$.
  Both points are not below the line supported by $e$ and
  the same holds for $g$ and $g_r$ with respect to $e_r$.
  It follows that segments $\overline{p_\ell g}$, $\overline{gp_r}$,
  and $\overline{p_r g_r}$ form an $x$-monotone convex chain that
  is nowhere below $T$. Hence,  $g_r$ sees $p_\ell$.

  Thus, any point $p\in\V(g)$ to the left of $g$
  is also seen by $g_r$, a contradiction to $g$ being a right-guard.
\end{proof}



The next theorem shows that the set $U$ as defined in Equation~\eqref{eq:u} contains all guard candidates necessary for a minimum-cardinality guard cover of $T$.
So even if we are allowed to place guards anywhere on $T$, we only need those in $U$ and thus have discretized the problem.

\begin{theorem}\label{thm:g}
	Let $T$ be a terrain, $C \subset T$ a finite guard cover of $T$, possibly of minimum cardinality, and consider $U$ as defined in Equation~\eqref{eq:u}.
	Then there exists a guard cover $C' \subseteq U$ of $T$ with $|C'| = |C|$.
\end{theorem}

\begin{proof}
	We iteratively replace a guard $g \in C \setminus U$ by one in $U$ until $C \subseteq U$.
	This maintains the cardinality of $C$, thus constructing the set $C'$ as claimed.

	Should $g$ be neither left- nor right-guard, it can be moved to a neighboring point in $U$ by Lemma~\ref{lem:single-interval-vis}.
	If, on the other hand, $g$ is only a left-, but not a right-guard (or vice versa), it can be moved to its left (right) neighbor in $U$ as shown in Lemma~\ref{lem:left-guard-to-u} and Corollary~\ref{cor:right-guard-to-u}.
	Lemma~\ref{lem:leftright-guard-to-u} states that $g$ cannot be a left- and a right-guard at the same time.
\end{proof}

\subsection{Complete Discretization}
\label{sec:compl}

In this section, we formulate our key result.
Let $\opt(G,W)$ denote the cardinality of an optimal, \ie, minimum-cardinality, solution for $\tgp(G,W)$.

\begin{theorem}\label{th:opt}
Let $T$ be a terrain, $U$ and $W(U)$ as defined in Equations~\eqref{eq:u} and~\eqref{eq:w}.
Then:
If $C$ is an optimal solution of $\tgp(U, W(U))$, \ie, $|C|=\opt(U,W(U))$, then $C$ is also an optimal solution of $\tgp(T,T)$, \ie, $\opt(T,T) = |C| =  \opt(U,W(U))$.
\end{theorem}

\begin{proof}
We have $\opt(T,T) \leq \opt(U,T)$.
Theorem~\ref{thm:g} states that for an optimal guard cover $C''$ of $\tgp(T,T)$ there exists a guard cover $C'\subseteq U$ with $|C'| = |C''|$, \ie,
\begin{equation}
\opt(T,T) = |C''| = |C'| \geq \opt(U,T).
\end{equation}
This yields $\opt(T,T) = \opt(U,T)$.

If we are given an optimal solution $C$ to $\tgp(U,W(U))$, with $\opt(U, W(U))$ guards, $C$ is a feasible solution for $\tgp(U,T)$ according to Theorem~\ref{thm:w}. As $|C| = \opt(U,W(U)) \leq \opt(U,T)$ and $C$ is feasible for $\tgp(U,T)$, we have $\opt(U,W(U)) = \opt(U,T)$ which concludes the proof.
\end{proof}

\begin{obs}\label{obs:card}
Observations~\ref{obs:witcard} and~\ref{obs:guardcard} yield:
The set of guard candidates $U$ has cardinality $\bigO(n^2)$, the finite witness set $W(U)$ has cardinality $\bigO(n^3)$.
\end{obs}

\section{The PTAS}
\label{sec:ptas}

In this section we combine the PTAS of Gibson et al.~\cite{gkkv-aasftg-09}
with the results from Section~\ref{sec:discretization} to obtain a PTAS for $\tgp(T,T)$.

Let us first formulate the result of Gibson et al., who presented a PTAS for the discrete terrain guarding problem, in our notation:

\begin{lemma}[PTAS by Gibson et al.~\cite{gkkv-aasftg-09}]\label{le:gibson}
Let $T$ be a terrain, and $G,W \subset T$ finite sets of guard candidates and points to be guarded. Then there exists a polynomial time approximation scheme for $\tgp(G,W)$. That is, there exists a polynomial time algorithm that returns a subset $C\subseteq G$ with $W\subseteq \V(C)$, such that $|C|  \leq (1+\epsilon)\cdot \opt(G,W) \; \forall \epsilon > 0$, where $\opt(G,W)$ denotes the optimal solution for $\tgp(G,W)$.
\end{lemma}

We can now easily combine Theorem~\ref{th:opt} and Lemma~\ref{le:gibson} for a PTAS for the continuous TGP:

\begin{theorem}\label{th:ptas}
Let $T$ be a terrain. Then there exists a polynomial time approximation scheme for $\tgp(T,T)$, the continuous terrain guarding problem. That is, there exists a polynomial time algorithm that returns a subset $C\subset T$ with $\V(C)=T$, such that $|C|  \leq (1+\epsilon)\cdot \opt(T,T) \; \forall \epsilon > 0$, where $\opt(T,T)$ denotes the optimal solution for $\tgp(T,T)$.
\end{theorem}

\begin{proof}
Using Equations~\eqref{eq:u} and~\eqref{eq:w} we determine the sets $U$ and $W(U)$ for the terrain $T$, two finite subsets of $T$. Given an arbitrary $\epsilon > 0$  we can compute a set $C\subseteq U \subset T$ with $\V(C)=T$ such that:
\begin{equation}\label{eq:eps}
 |C|  \leq (1+\epsilon)\cdot \opt(U,W(U))
\end{equation}
using the PTAS of Gibson et al., Lemma~\ref{le:gibson}.
Moreover, Theorem~\ref{th:opt} yields:
\begin{equation}\label{eq:id}
\opt(U,W(U)) = \opt(T,T).
\end{equation}
Combining Equations~\eqref{eq:eps} and~\eqref{eq:id} we obtain
\begin{equation}
	\begin{aligned}
		|C| & \leq (1 + \epsilon) \cdot \opt(U, W(U)) \\
		    & =    (1 + \epsilon) \cdot \opt(T, T)
	\end{aligned}
\end{equation}
as claimed.
\end{proof}

\section{Exact Solutions}
\label{sec:exact}

Let $T$ be a terrain, and $U$ and $W(U)$ be defined as in Equations~\eqref{eq:u} and~\eqref{eq:w}.
Our discretization allows us to formulate the following integer program (IP)
to find an exact solution of $\tgp(T,T)$ by modeling guard candidates as binary variables and witnesses as constraints:
\begin{alignat}{3}
	\textnormal{min}    & \sum_{g \in U} x_g \label{eq:ip-begin} \\
	\textnormal{s.\,t.} & \sum_{g \in \V(w) \cap U} x_g \geq 1 & \quad & \forall w \in W(U) \\
	                    & x_g \in \{0, 1\}                     & \quad & \forall g \in U \label{eq:ip-end}
\end{alignat}
The IP in \eqref{eq:ip-begin}~--~\eqref{eq:ip-end} paves the way for an exact, while not polytime, algorithm for the continuous TGP.

\section{Conclusion}
\label{sec:conclusion}

We showed that if we want to solve $\tgp(G,T)$ with a finite guard candidate set $G$, we can find a finite witness set $W(G)$, such that a solution of $\tgp(G,W(G))$ is feasible for $\tgp(G,T)$.
Our main result is the construction of a finite set of guard candidates $U$ of size $\bigO(n^2)$ such that $U$ admits an optimal cover of $T$.

Combining the two discretizations we concluded that an optimal solution of $\tgp(U, W(U))$ is also an optimal solution of $\tgp(T,T)$.

A polynomial time approximation scheme (PTAS) for $\tgp(T,T)$ using a former PTAS~\cite{gkkv-aasftg-09} for the discrete terrain guarding problem $\tgp(G,W)$ immediately follows.
Moreover, we formulate $\tgp(T,T)$ as an integer program (IP), yielding exact solutions.


\small
\bibliographystyle{abbrv}
\balance
\bibliography{bibliography}

\end{document}